\documentclass[11pt]{article}
\usepackage[utf8]{inputenc}
\usepackage[T1]{fontenc}
\usepackage[margin=1in]{geometry}
\usepackage{graphicx}
\graphicspath{{./}{../results/}}
\usepackage{booktabs}
\usepackage{amsmath,amssymb,amsthm}
\usepackage{authblk}
\usepackage[hidelinks]{hyperref}

\newcommand{\E}{\mathbb{E}}
\renewcommand{\Pr}{\mathbb{P}}
\newtheorem{proposition}{Proposition}
\newtheorem{lemma}{Lemma}
\newtheorem{corollary}{Corollary}
\newtheorem{remark}{Remark}
\title{How Much Do RF Drone Benchmarks Overstate?\\ A Controlled Study and Theory of Data Leakage in UAV Signal Identification}
\author{David Shulman\thanks{Email: david.shulman.research@gmail.com}}
\affil{\small\texttt{github.com/shulm/spectrahawk}}
\date{\today}
\begin{document}
\maketitle
\begin{abstract}
Radio-frequency (RF) sensing is a central modality for counter--unmanned-aerial-system (counter-UAS) defence, exploiting the control, telemetry, and video links between a drone and its operator. Reported accuracies for RF-based drone detection and identification are frequently very high, yet they are often produced by cross-validation that splits a small number of continuous recordings at the level of short segments, allowing near-duplicate slices of the same recording into both training and test partitions. We study this data-leakage pathology in depth, contributing both theory and measurement. We formalise the optimism of segment-level cross-validation and show, via Cover's function-counting theorem, that a classifier can memorise the recording-to-label map exactly while the number of independent recordings $R$ is small relative to the feature dimension $d$ (specifically while $2R \lesssim d$), so that naive accuracy approaches $1$ and the inflation gap approaches $1-\mathrm{ACC}^\star$, where $\mathrm{ACC}^\star$ is the Bayes accuracy; the inflation eases only as $R$ grows past this separability threshold. A controlled synthetic experiment (10 seeds) confirms the predicted curves: naive balanced accuracy rises from the Bayes level toward $1.0$ as a recording-specific nuisance grows, while honest recording-grouped evaluation declines to chance, with a gap reaching $\approx 0.5$. On the public DroneRF dataset, pooled leave-one-recording-out cross-validation shows drone type identification (AR vs.\ Bebop) collapsing from a naive macro-$F_1$ of $0.74$ to $0.46$ -- the two-class chance level -- and a leakage-pathway ablation attributes essentially all of the inflation to segment-level leakage. We conclude that RF drone-identification benchmarks must use recording-grouped evaluation and enough independent captures, and we release all code and a synthetic generator for reproducibility.
\end{abstract}

\section{Introduction}
Most consumer and commercial drones remain in continuous radio contact with their operators, producing structured RF emissions -- frequency-hopping control links, wideband video downlinks, and telemetry -- that a passive receiver can exploit for detection and identification. RF sensing offers range and signal richness and complements acoustic methods, which remain effective when a drone is radio-silent. A persistent obstacle to trustworthy progress, however, is evaluation rigour. Public RF drone datasets typically consist of a small number of continuous recordings, each cut into many short segments. When such segments are partitioned at random, segments from the \emph{same} recording -- sharing the same channel, the same ambient background, and the same device at a near-identical operating point -- appear in both the training and test partitions. A classifier can then succeed by recognising the recording rather than the drone, and the reported score reflects memorisation rather than generalisation.

This paper quantifies and \emph{explains} that effect. Beyond demonstrating leakage empirically, we ask: \emph{how large should the inflation be, and on what does it depend?} We give a theory that answers this for a transparent generative model, and we confirm it both in controlled simulation and on real data. Our contributions are:
\begin{enumerate}
\item a \textbf{formal account} of segment-level cross-validation optimism, decomposing it into a term governed by the recoverability of the recording identity (Section~\ref{sec:formal});
\item a \textbf{theoretical prediction} of the leakage-inflation curves: using Cover's theorem, the naive split memorises recordings while $2R \lesssim d$, so naive accuracy $\to 1$ and the gap $\to 1-\mathrm{ACC}^\star$, easing as $R$ grows (Section~\ref{sec:theory});
\item a \textbf{controlled synthetic experiment} (10 seeds) that confirms the predicted dependence on nuisance strength and on $R$ (Section~\ref{sec:controlled});
\item a \textbf{real-world case study} on DroneRF with pooled leave-one-recording-out cross-validation and a leakage-pathway ablation (Section~\ref{sec:dronerf}).
\end{enumerate}
The practical stakes are concrete. A counter-UAS operator who procures a detector advertised at $99\%$ identification accuracy, but whose evaluation leaked recordings, fields a system that performs near chance against drones flown in any new environment -- a failure mode that surfaces only after deployment. The discipline we advocate is therefore not pedantry but a procurement and operational-safety concern: the headline number and the fielded number can differ by the full inflation gap we quantify.

All experiments are released as an open, reproducible pipeline with a synthetic generator, so the central results run without any download.

\section{Related work}
RF-based drone detection and identification commonly transform captured segments into spectral or time--frequency features and apply classical or deep classifiers; reported accuracies on datasets such as DroneRF~\cite{dronerf} frequently exceed $95\%$. Data leakage -- the contamination of evaluation by information available at training time -- is a recognised and widespread cause of over-optimistic, non-reproducible machine-learning results~\cite{kapoor}, and grouped (subject-wise) cross-validation is the standard remedy in fields such as biosignal analysis. A closely related phenomenon is \emph{shortcut learning}~\cite{geirhos}, in which models exploit spurious features that do not generalise. Our separability analysis rests on Cover's classical function-counting theorem for linear dichotomies~\cite{cover}. We make these ideas concrete for RF drone recognition: we exhibit the mechanism in a controlled setting, predict its magnitude, quantify it on real data, and trace its sources. Grouped (subject- or recording-wise) cross-validation is the accepted safeguard in adjacent fields such as EEG and audio event detection, where repeated samples from one subject or session otherwise leak; our analysis supplies a quantitative account of \emph{how much} is at stake as a function of dataset structure.

\section{RF signal model and features}
\label{sec:rf}
\subsection{Emission and channel model}
A drone link is a passband signal represented at baseband by its in-phase and quadrature components, $x(t)=I(t)+jQ(t)$. Two structures dominate. A frequency-hopping control link visits a sequence of narrow channels,
\begin{equation}
x_{\mathrm{ctrl}}(t)=\sum_{h} g\!\left(t-hT_h\right)\,e^{\,j2\pi f_h t},
\end{equation}
with per-hop carrier $f_h$ drawn from a hop set, dwell $T_h$, and pulse shape $g$. A wideband video downlink is well modelled as multi-carrier (OFDM),
\begin{equation}
x_{\mathrm{vid}}(t)=\sum_{k} c_k\,e^{\,j2\pi k\,\Delta f\,t},
\end{equation}
occupying a broad contiguous band. The receiver observes
\begin{equation}
y(t) = (h_r * x)(t) + w_r(t),
\end{equation}
where the propagation channel $h_r$ and the ambient interference $w_r$ (other Wi-Fi/Bluetooth emitters, thermal noise) are \emph{specific to a recording} $r$. This recording-specific term is the physical origin of the nuisance that drives leakage: two segments of the same flight share $h_r$ and the statistics of $w_r$.

\subsection{Spectral representation and features}
From the short-time Fourier transform $X(m,k)$ we form the power spectrum $S(f)=\E_m|X(m,f)|^2$ and summarise it with interpretable, scale-aware descriptors. With normalised spectrum $p(f)=S(f)/\sum_f S(f)$,
\begin{align}
\text{band power} &: \quad P_B=\textstyle\sum_{f\in B} S(f), &
\text{centroid} &: \quad \mu=\textstyle\sum_f f\,p(f),\\
\text{flatness} &: \quad \mathrm{SFM}=\frac{\exp\!\big(\frac{1}{F}\sum_f \log S(f)\big)}{\frac{1}{F}\sum_f S(f)}, &
\text{$99\%$ bandwidth} &: \quad B_{99}=\min\{|B|: \textstyle\sum_{f\in B}S(f)\ge 0.99\,\textstyle\sum_f S(f)\}.
\end{align}
A frequency-hopping control link yields scattered narrow peaks (low flatness, narrow instantaneous occupancy); a video downlink yields a wide, flat band (high $B_{99}$). These features form the input to the baseline classifier.

\section{Problem formulation}
\label{sec:formal}
\subsection{Tasks and metrics}
\emph{Detection} is the binary decision of drone presence, characterised by the probability of detection $P_d$ and false-alarm rate $P_{fa}$; sweeping the threshold traces the receiver operating characteristic with area $\mathrm{AUC}=\int_0^1 P_d\,dP_{fa}$, and we report $P_d$ at a fixed $P_{fa}=1\%$. \emph{Identification} is multi-class, summarised by the macro-averaged $F_1$, $\mathrm{macro\text{-}}F_1=\frac{1}{C}\sum_{c}F_{1,c}$, computed once over predictions pooled across folds with fixed labels. To avoid imbalance artefacts we use balanced accuracy $\mathrm{bACC}=\frac{1}{C}\sum_{c}\mathrm{recall}_c$, for which chance is $1/C$.

\subsection{Leakage as estimator optimism}
Let a dataset be $\mathcal{D}=\{(\mathbf{x}_i,y_i,g_i)\}_{i=1}^{n}$, where $g_i$ indexes the continuous \emph{recording} from which segment $i$ was cut. For a learned predictor $\hat f$ the quantity of interest is the population risk $R(\hat f)=\E_{(\mathbf{x},y)\sim P}[\,\ell(\hat f(\mathbf{x}),y)\,]$ on \emph{new recordings}. A cross-validation scheme produces an estimate $\hat R$; its \emph{optimism} is $\mathrm{opt}=\E[\,R(\hat f)-\hat R\,]$. Two schemes differ only in how folds respect $g$:
\begin{itemize}
\item \textbf{Naive} (segment-level): folds ignore $g$, so for most test segments their recording also appears in training.
\item \textbf{Grouped}: folds are group-disjoint, $\mathcal{G}_{\mathrm{train}}\cap\mathcal{G}_{\mathrm{test}}=\varnothing$.
\end{itemize}
Because each recording carries a single label, the conditional law factorises as $\Pr(y\mid \mathbf{x},g)$, and a sufficiently expressive model trained under the naive scheme can exploit $\Pr(y\mid g)$, which is degenerate (a point mass), achieving near-zero \emph{apparent} risk. The grouped scheme withholds $g$ at test time and forces reliance on $\Pr(y\mid \mathbf{x})$. The optimism of the naive estimator is therefore approximately
\begin{equation}
\mathrm{opt}_{\mathrm{naive}} \;\approx\; R^\star \;-\; \underbrace{R_{\mathrm{mem}}}_{\approx\,0\ \text{when $g$ is recoverable}} \;=\; R^\star,
\label{eq:opt}
\end{equation}
where $R^\star$ is the Bayes risk on unseen recordings and $R_{\mathrm{mem}}$ is the apparent risk of the recording-memorising solution. Equation~\eqref{eq:opt} says the inflation is large exactly when (i) the true task is hard ($R^\star$ large) and (ii) the recording identity $g$ is recoverable from $\mathbf{x}$. Section~\ref{sec:theory} makes (ii) precise.

\subsection{Background: the optimism of cross-validation}
For a loss $\ell$ and a learning algorithm trained on a sample, the apparent (in-sample) error and the out-of-sample error differ by the \emph{optimism}, whose expectation is positive whenever evaluation reuses information from training. Ordinary $k$-fold cross-validation controls this when examples are exchangeable and independent. The grouped structure of segmented recordings violates exchangeability: examples within a recording are strongly dependent and the recording's label is constant, so a fold holding part of a recording at training time and the rest at test time measures within-recording fit rather than across-recording generalisation. Grouped cross-validation restores independence \emph{at the recording level}, at the cost of fewer effective evaluation units -- which is exactly why the number of independent recordings, not the number of segments, governs both the bias and the variance of the estimate. This is the lens through which the next section's results should be read.

\section{A theory of leakage optimism}
\label{sec:theory}
\subsection{Generative model}
We instantiate the above in feature space with known ground truth. We draw $R$ recordings per class; recording $r$ has label $y_r\in\{0,1\}$ and a per-recording nuisance $\boldsymbol{\nu}_r$. Each of its $S$ segments is
\begin{equation}
\mathbf{x} = \mathbf{m}_{y_r} + \boldsymbol{\nu}_r + \boldsymbol{\varepsilon}, \qquad \boldsymbol{\varepsilon}\sim\mathcal{N}(\mathbf{0},\sigma^2\mathbf{I}_d),
\label{eq:gen}
\end{equation}
where the class mean $\mathbf{m}_y$ separates the classes \emph{only} along axis $0$ by $\Delta$ (class $0$ at $-\Delta/2$, class $1$ at $+\Delta/2$), and the nuisance is confined to the non-discriminative axes, $\boldsymbol{\nu}_r\sim\mathcal{N}(\mathbf{0},\lambda^2\mathbf{I})$ on axes $1,\dots,d{-}1$ and $0$ on axis $0$. Thus the only feature that generalises across recordings is axis $0$. The nuisance is class-independent and changes no true separability; it exists only to make a recording identifiable, as a real channel $h_r$ and ambient $w_r$ do.

\begin{proposition}[Bayes accuracy]
\label{prop:bayes}
Under \eqref{eq:gen}, the optimal accuracy attainable on a segment from an unseen recording is $\mathrm{ACC}^\star=\Phi(\Delta/2\sigma)$, achieved by thresholding axis $0$; no function of the nuisance axes improves it.
\end{proposition}
\begin{proof}[Sketch]
Axis $0$ is $\mathcal{N}(\pm\Delta/2,\sigma^2)$ by class and is independent of $\boldsymbol{\nu}_r$; the Bayes rule for two equal-variance Gaussians separated by $\Delta$ has error $\Phi(-\Delta/2\sigma)$. The nuisance axes are identically distributed across classes (class-independent), hence carry no information about $y$ for an unseen recording.
\end{proof}

\subsection{When can the naive split memorise recordings?}
Under the naive scheme, segments of a training recording are visible, so the learner may instead solve the proxy task of mapping each segment to its recording's label using the nuisance axes. Whether this proxy is realisable is a question of linear separability of the $2R$ recording centroids $\{\boldsymbol{\nu}_r\}$ -- in general position in $\mathbb{R}^{d-1}$ -- by their labels.

\begin{lemma}[Cover~\cite{cover}]
\label{lem:cover}
For $N$ points in general position in $\mathbb{R}^{p}$ and a fixed labelling, the probability that the labelling is linearly separable (through the origin) is
\begin{equation}
C(N,p)=2^{\,1-N}\sum_{k=0}^{p-1}\binom{N-1}{k},
\end{equation}
which satisfies $C(N,p)\approx 1$ for $N\le p$ and $C(2p,p)=\tfrac12$.
\end{lemma}

\begin{proposition}[Predicted inflation]
\label{prop:infl}
Consider \eqref{eq:gen} in the regime $\lambda\gg\sigma$ (nuisance dominant). Then:
\begin{enumerate}
\item[(a)] With probability $C(2R,\,d{-}1)$, a linear classifier on the nuisance axes separates the $2R$ recordings by label; whenever it does, the naive cross-validated accuracy is $\approx 1$ (every test segment is routed by its recording), so the optimism gap is $\approx 1-\mathrm{ACC}^\star$.
\item[(b)] As $R$ increases past the threshold $2R\approx d{-}1$, separability -- hence exact memorisation -- fails with growing probability, and naive accuracy declines toward an SNR-limited value.
\item[(c)] The grouped accuracy is, in expectation, at most $\mathrm{ACC}^\star$ (Proposition~\ref{prop:bayes}) and tends to chance when finite training recordings make axis $0$ unrecoverable amid the $d{-}1$ dominant nuisance axes.
\end{enumerate}
\end{proposition}
\begin{proof}[Sketch]
(a) When $\lambda\gg\sigma$, segments cluster tightly around their recording centroids $\boldsymbol{\nu}_r$; a hyperplane that separates the centroids by label classifies essentially all segments correctly, and Lemma~\ref{lem:cover} gives the probability such a hyperplane exists for the fixed labelling. The true test accuracy is unchanged at $\mathrm{ACC}^\star$, so the gap is $1-\mathrm{ACC}^\star$. For (b), once $2R$ exceeds $d{-}1$ the probability $C(2R,d{-}1)$ that the fixed labelling of the centroids is linearly separable decays from one toward zero (Lemma~\ref{lem:cover}, with $C(2(d{-}1),d{-}1)=\tfrac12$), so exact memorisation by a linear rule becomes impossible for a growing fraction of draws and the naive accuracy must fall from its ceiling toward whatever the residual axis-$0$ signal supports. For (c), the grouped test set shares no recording with training, so by Proposition~\ref{prop:bayes} no nuisance-based rule can exceed chance and the best attainable accuracy is $\mathrm{ACC}^\star$; moreover, estimating the $\Delta$-separated direction from training data is itself ill-conditioned when $d{-}1$ axes of variance $\lambda^2\gg\sigma^2$ dominate the empirical covariance and only a few independent recordings are available, so the realised grouped accuracy lies between chance and $\mathrm{ACC}^\star$ and approaches chance as $\lambda$ grows.
\end{proof}

\begin{corollary}[Few recordings are the worst case]
\label{cor:few}
The leakage inflation is maximised when the number of independent recordings is small relative to the feature dimension ($2R\lesssim d$) -- precisely the regime of public RF drone datasets, which contain only a handful of independent captures per class.
\end{corollary}

\begin{remark}
The deployed baseline is a (nonlinear) Random Forest, for which the linear bound of Lemma~\ref{lem:cover} is conservative: nonlinear models can memorise recordings even more readily, so the predicted naive inflation is a lower bound on what flexible models achieve. The theory therefore explains, rather than merely describes, the empirical curves below.
\end{remark}

\subsection{Detection ROC under the model}
For binary detection of presence using axis $0$, the score is Gaussian under each hypothesis, $z\mid y \sim \mathcal{N}(\pm\Delta/2,\sigma^2)$. The honest ROC is the standard shifted-Gaussian curve $P_d=\Phi\!\big(\Phi^{-1}(P_{fa})+\Delta/\sigma\big)$, with area
\begin{equation}
\mathrm{AUC}^\star=\Phi\!\left(\frac{\Delta}{\sigma\sqrt2}\right).
\end{equation}
The same axis-$0$ separation that fixes the Bayes accuracy (Proposition~\ref{prop:bayes}) thus fixes the honest detection ROC; under naive evaluation, recording memorisation again drives AUC toward $1$ independently of $\Delta$, by the mechanism of Proposition~\ref{prop:infl}. Detection and identification therefore inflate through the same channel.

\section{Controlled experiment}
\label{sec:controlled}
We instantiate \eqref{eq:gen} with $d=20$, $S=50$, $\sigma=1$, and $\Delta$ set so $\mathrm{ACC}^\star\approx0.65$. We compare a naive partition (\texttt{StratifiedKFold} over segments) with an honest partition (\texttt{StratifiedGroupKFold} keyed on recording), using a Random Forest (primary) and a Logistic Regression (corroboration), reporting balanced accuracy as mean $\pm$ 95\% CI over 10 dataset seeds.

Figure~\ref{fig:lam} and Table~\ref{tab:lam} sweep the nuisance $\lambda$ at $R=8$. As predicted by Proposition~\ref{prop:infl}(a), naive accuracy climbs monotonically toward $1.0$ as $\lambda$ grows, while honest evaluation declines to chance; the gap reaches $\approx0.5\approx 1-\mathrm{ACC}^\star$ minus a finite-sample residual. The Random Forest's honest accuracy approaches chance and can dip slightly below in individual finite-recording draws (shortcut learning~\cite{geirhos}); the regularised Logistic Regression declines toward chance while retaining only a weak vestige of the axis-$0$ signal.

\begin{table}[t]\centering
\caption{Controlled experiment: balanced accuracy (mean $\pm$ 95\% CI, 10 seeds) versus nuisance $\lambda$ at $R=8$. True Bayes $\mathrm{ACC}^\star\approx0.65$.}
\label{tab:lam}
\begin{tabular}{lccc}
\toprule
$\lambda$ & RF naive & RF grouped & LR grouped \\
\midrule
0.0 & $0.628\pm0.013$ & $0.617\pm0.010$ & $0.632\pm0.009$ \\
0.5 & $0.735\pm0.015$ & $0.600\pm0.018$ & $0.589\pm0.025$ \\
1.0 & $0.903\pm0.014$ & $0.566\pm0.040$ & $0.574\pm0.047$ \\
2.0 & $0.995\pm0.002$ & $0.515\pm0.068$ & $0.584\pm0.076$ \\
4.0 & $1.000\pm0.000$ & $0.497\pm0.063$ & $0.584\pm0.095$ \\
8.0 & $1.000\pm0.000$ & $0.492\pm0.071$ & $0.571\pm0.101$ \\
\bottomrule
\end{tabular}
\end{table}

\begin{figure}[t]\centering
\IfFileExists{leakage_vs_lambda.png}{\includegraphics[width=0.62\linewidth]{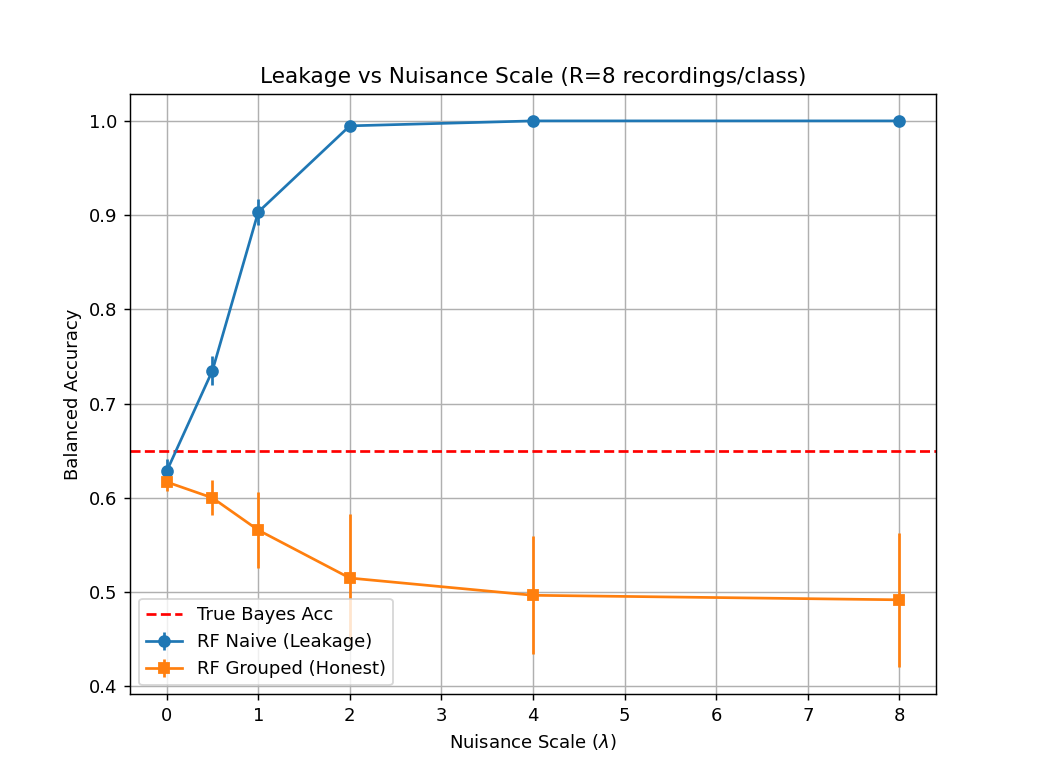}}{\fbox{\rule{0pt}{3.5cm}\rule{0.55\linewidth}{0pt}}}
\caption{Balanced accuracy versus nuisance $\lambda$ (10 seeds, 95\% CI). Naive cross-validation rises to $1.0$ by memorising recordings; honest grouped cross-validation declines to chance. The growing gap is the leakage inflation predicted by Proposition~\ref{prop:infl}.}
\label{fig:lam}
\end{figure}

Figure~\ref{fig:R} and Table~\ref{tab:R} sweep the number of recordings $R$ at fixed $\lambda=2$. Consistent with Proposition~\ref{prop:infl}(b) and Corollary~\ref{cor:few}, naive accuracy is highest at small $R$ (where $2R\lesssim d{-}1=19$, i.e.\ $R\lesssim9$, so recordings are linearly separable) and declines gently as $R$ grows past the threshold. Honest accuracy stays near chance at this nuisance level but its variance shrinks markedly (95\% CI from $\pm0.13$ at $R{=}2$ to $\pm0.03$ at $R{=}32$): more independent recordings stabilise the honest estimate even where the dominant nuisance prevents the weak true signal from being recovered.

\begin{table}[t]\centering
\caption{Controlled experiment: balanced accuracy (mean $\pm$ 95\% CI, 10 seeds) versus recordings $R$ per class, at $\lambda=2$.}
\label{tab:R}
\begin{tabular}{lcc}
\toprule
$R$ & RF naive & RF grouped \\
\midrule
2 & $0.999\pm0.002$ & $0.495\pm0.134$ \\
4 & $0.998\pm0.002$ & $0.608\pm0.083$ \\
8 & $0.995\pm0.002$ & $0.515\pm0.068$ \\
16 & $0.985\pm0.003$ & $0.502\pm0.048$ \\
32 & $0.972\pm0.003$ & $0.543\pm0.030$ \\
\bottomrule
\end{tabular}
\end{table}

\begin{figure}[t]\centering
\IfFileExists{leakage_vs_recordings.png}{\includegraphics[width=0.62\linewidth]{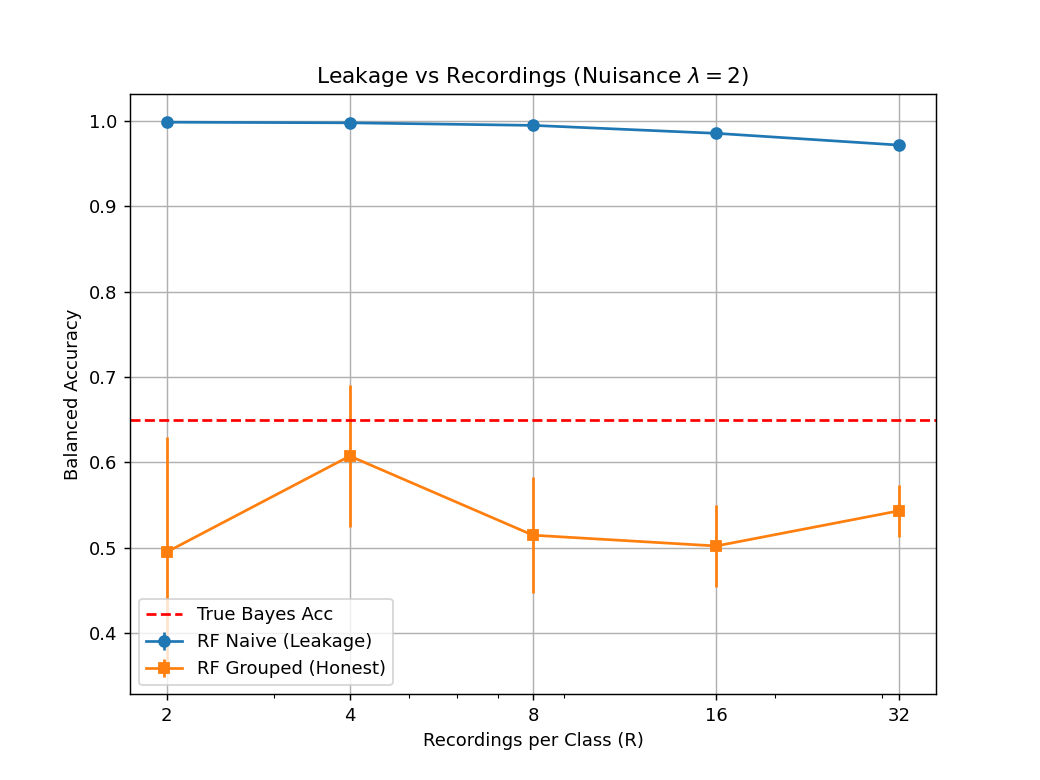}}{\fbox{\rule{0pt}{3.5cm}\rule{0.55\linewidth}{0pt}}}
\caption{Balanced accuracy versus number of recordings $R$ ($\lambda=2$, 10 seeds, 95\% CI). Naive inflation is largest below the Cover separability threshold ($2R\approx d$) and eases as $R$ grows; the honest estimate's variance shrinks with $R$.}
\label{fig:R}
\end{figure}

\subsection{The separability bound versus the data}
Figure~\ref{fig:cover} overlays the linear separability probability $C(2R,d{-}1)$ (Lemma~\ref{lem:cover}, $d=20$) with the empirical naive accuracy. The bound predicts that exact recording memorisation by a linear classifier is certain up to $R\approx 8$ and collapses beyond the threshold $2R\approx d$. The measured naive accuracy stays high well past this point, consistent with Remark~1 (Section~\ref{sec:theory}): the nonlinear Random Forest memorises recordings even where no linear separator exists, so the linear threshold is a \emph{conservative} lower bound on the leakage regime. The qualitative agreement -- inflation guaranteed below the threshold and persisting above it for flexible models -- is exactly what the theory predicts.

\begin{figure}[t]\centering
\IfFileExists{cover_vs_R.png}{\includegraphics[width=0.60\linewidth]{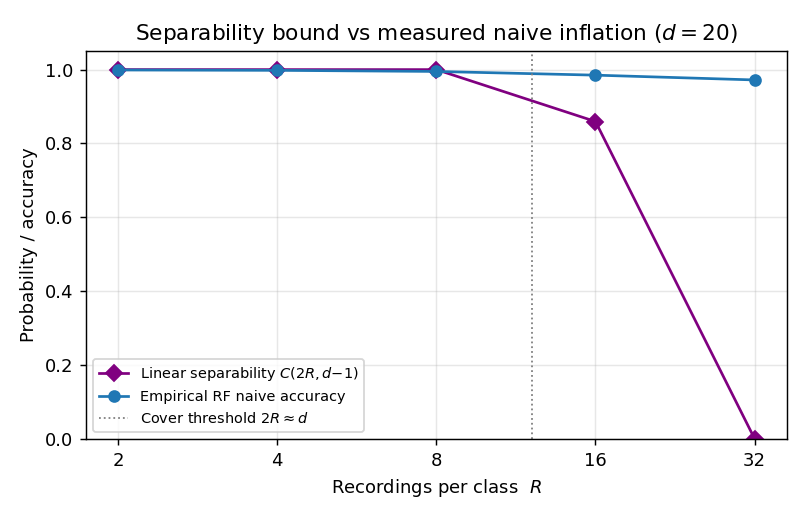}}{\fbox{\rule{0pt}{3.5cm}\rule{0.55\linewidth}{0pt}}}
\caption{Theory versus data. The linear separability probability $C(2R,d{-}1)$ (Cover) collapses past $2R\approx d$, but the nonlinear Random Forest's naive accuracy remains high, confirming that the linear threshold lower-bounds the leakage regime.}
\label{fig:cover}
\end{figure}

\section{Real-world case study: DroneRF}
\label{sec:dronerf}
\subsection{Dataset}
DroneRF~\cite{dronerf} provides $2.4$\,GHz recordings of three drones (AR, Bebop, Phantom) and a no-drone background. The band is captured by two receivers -- a low half and a high half -- \emph{simultaneously, from the same flight}; we merge the two band-halves into a single recording group. After merging, the number of \emph{independent} recordings is small: one (background), four (AR), four (Bebop), and one (Phantom). Hence a single background recording precludes a leakage-free detection split (detection is reported drone-side-grouped with residual background leakage), and honest identification is limited to AR vs.\ Bebop (Phantom, one recording, is excluded from the grouped metric). By Corollary~\ref{cor:few}, this few-recording regime is exactly where leakage inflation is largest.

\subsection{Leave-one-recording-out evaluation}
We evaluate with leave-one-recording-out cross-validation, pooling the out-of-fold predictions across folds and computing a single macro-$F_1$ with fixed labels (95\% CI by bootstrap over pooled predictions), since each held-out recording is a single class. Honest identification of AR versus Bebop collapses from a naive macro-$F_1$ of $0.74$ to $0.46$ -- the two-class chance level -- once whole recordings are held out (Table~\ref{tab:loro}). Detection AUC barely changes (the background still leaks, by necessity), and its $1\%$-false-alarm detection rate is highly variable across folds.

\begin{table}[t]\centering
\caption{DroneRF results: Type-ID pooled out-of-fold (95\% CI by bootstrap); Detection per-fold mean $\pm$ std.}
\label{tab:loro}
\begin{tabular}{llcc}
\toprule
Task & Metric & Naive & Honest \\
\midrule
Type-ID (AR vs Bebop) & macro-$F_1$ & $0.742\,[.728,.755]$ & $\mathbf{0.455\,[.439,.470]}$ \\
Detection (drone-side grp., & ROC-AUC & $0.985\pm0.005$ & $0.978\pm0.017$ \\
\ resid.\ bg.\ leakage) & $P_d$@$1\%$FA & $0.764\pm0.279$ & $0.719\pm0.275$ \\
\bottomrule
\end{tabular}
\end{table}

\subsection{Leakage-pathway ablation}
Table~\ref{tab:abl} tightens the grouping in three steps (pooled out-of-fold macro-$F_1$). Segment-level evaluation reports $0.739$; grouping by file (band-halves separate) drops it to $0.527$; merging the two band-halves of a flight changes it only marginally to $0.524$. Essentially all of the inflation is \emph{segment-level} leakage within a recording; the band-half merge contributes little under this metric. This is an attribution based on the dataset's known structure (H and L are simultaneous captures of one flight), not a randomized intervention.

\begin{table}[t]\centering
\caption{Leakage-pathway ablation on DroneRF (AR vs Bebop), pooled out-of-fold macro-$F_1$.}
\label{tab:abl}
\begin{tabular}{lc}
\toprule
Grouping level & macro-$F_1$ \\
\midrule
L0 -- segment-level (no grouping) & 0.739 \\
L1 -- file-level (band-halves separate) & 0.527 \\
L2 -- recording-level (band-halves merged) & 0.524 \\
\bottomrule
\end{tabular}
\end{table}

\section{Discussion}
Theory and measurement agree. The optimism decomposition \eqref{eq:opt} predicts large inflation when the task is hard and the recording is identifiable; the Cover analysis (Proposition~\ref{prop:infl}) makes the second condition quantitative and ties it to the ratio of independent recordings to feature dimension; the controlled experiment confirms the $\lambda$- and $R$-dependence; and DroneRF, with only a few independent recordings per class, exhibits exactly the predicted collapse to chance under honest evaluation. As a worked estimate, DroneRF identification has $R=4$ recordings per class, so $2R=8$; with a feature dimension of order $d\sim20$ the separability index $2R/d\approx0.4$ places the dataset deep in the memorisation regime, $C(2R,d{-}1)\approx1$. The theory thus predicts that naive evaluation can recover the recording-to-label map almost perfectly and inflate accuracy by up to $1-\mathrm{ACC}^\star$ above the honest level; the measured jump from a chance-level honest macro-$F_1$ of $0.46$ to a naive $0.74$ is consistent with this once the nonlinearity of the Random Forest (Remark~1) is accounted for. Three practical recommendations follow. First, evaluate RF drone classifiers with \emph{recording-grouped} cross-validation, and merge simultaneous captures of one flight (e.g.\ band-halves) into a single group. Second, report naive and grouped scores side by side so that any inflation is visible. Third -- the message of Corollary~\ref{cor:few} -- ensure the dataset contains enough \emph{independent} recordings that grouped estimates are both honest and stable; a benchmark with a handful of captures per class cannot support a confident identification claim, however high its naive accuracy.

\paragraph{Limitations.} The controlled model is a feature-space idealisation; real features are nonlinear functions of IQ with correlated nuisance, though this only strengthens memorisation (Remark~1). The DroneRF detection result remains drone-side-grouped because the dataset has a single background recording. We attempted to escalate to the richer DroneDetect dataset~\cite{dronedetect}, which has multiple captures per model, but it is access-restricted and lacks a drone-free background class. The grouped estimates on so few recordings are necessarily high-variance, which is itself part of the message.

\section{Implications for existing RF benchmarks}
Corollary~\ref{cor:few} bears directly on the public RF datasets we engaged with. DroneRF has at most four independent recordings per drone after band-halves are merged -- squarely in the few-recording regime where inflation is largest. A widely used noisy-RF classification benchmark is distributed as pre-sliced segments with \emph{no} recoverable recording identity, so grouped evaluation is impossible by construction and any reported accuracy is, in our terms, an upper bound dominated by leakage. The multi-capture DroneDetect set is better structured but is access-restricted and lacks a drone-free background class. In each case the practical conclusion is the same: high naive accuracy is not evidence of drone-discriminative learning unless accompanied by recording-grouped scores and a sufficient count of independent captures.

\section{A recommended evaluation protocol}
The analysis yields a concrete protocol for RF drone benchmarks.
\begin{enumerate}
\item \textbf{Identify the recording.} From filenames or metadata, determine the continuous capture each segment belongs to, and treat simultaneous sub-captures of one event (e.g.\ band-halves) as a single recording.
\item \textbf{Group, do not shuffle.} Partition by recording (\texttt{GroupKFold} or leave-one-recording-out), never by segment.
\item \textbf{Pool, then score.} For leave-one-recording-out with single-class folds, pool out-of-fold predictions and compute one metric with fixed labels; report a bootstrap confidence interval.
\item \textbf{Report both.} Always present naive and grouped scores side by side; their gap is the leakage estimate.
\item \textbf{Check the regime.} Report the number of independent recordings per class; if $2R\lesssim d$ (Corollary~\ref{cor:few}), treat identification claims as unsupported regardless of accuracy.
\end{enumerate}

\section{Reproducibility}
All experiments are released at \texttt{github.com/shulm/spectrahawk}: the synthetic generator that produces the controlled study with no download, deterministic seeds for every table, the DroneRF loader with recording-level grouping (band-halves merged), and the leave-one-recording-out and ablation scripts. The synthetic figures and Tables~\ref{tab:lam}--\ref{tab:R} reproduce exactly from a fixed seed; the DroneRF tables reproduce given the public dataset.

\section{Conclusion}
We gave a theory and a controlled-plus-real confirmation of data leakage in RF-based drone identification. Naive cross-validation can report near-perfect accuracy that is pure recording memorisation, with an inflation that our analysis predicts from the number of independent recordings relative to the feature dimension, and that we measure at up to $\approx0.5$ in simulation and as a collapse from $0.74$ to chance on DroneRF. The pitfall, and its remedy, generalise across drone-signature modalities. All code, figures, and the synthetic generator are released for reproduction.

\end{document}